\documentclass[smallextended,referee,envcountsect]{svjour3} 
\smartqed 
\usepackage{graphicx}

\usepackage{amsmath}
\usepackage{algorithmic}
\usepackage{amssymb}
\newcommand*{\field}[1]{\mathbb{#1}}%
\newtheorem{assumption}{Assumption}

\usepackage{bbm}
\begin{document}

\title{A Decentralized Multi-Objective Optimization Algorithm}

\author{Maude J. Blondin   \and  Matthew Hale }

\institute{Maude J. Blondin, Corresponding author \at
             Université de Sherbrooke \\
              Sherbrooke, Canada\\
              maude.josee.blondin@usherbrooke.ca
           \and
            Matthew Hale \at
              University of Florida \\
              Gainesville, Florida\\
              matthewhale@ufl.edu
}

\date{Received: date / Accepted: date}
%The correct dates will be entered by the editor.

\maketitle

\begin{abstract}
During the past two decades, multi-agent optimization problems have drawn increased attention from the research community. When multiple objective functions are present among agents, many works optimize the sum of these objective functions. However, this formulation implies a decision regarding the relative importance of each objective function. In fact, optimizing the sum is a special case of a multi-objective problem in which all objectives are prioritized equally. In this paper, a distributed optimization algorithm that explores Pareto optimal solutions for non-homogeneously weighted sums of objective functions is proposed. This exploration is performed through a new rule based on agents' priorities that generates edge weights in agents' communication graph. These weights determine how agents update their decision variables with information received from other agents in the network. Agents initially disagree on the priorities of the objective functions though they are driven to agree upon them as they optimize. As a result, agents still reach a common solution. The network-level weight matrix is (non-doubly) stochastic, which contrasts with many works on the subject in which it is doubly-stochastic. New theoretical analyses are therefore developed to ensure convergence of the proposed algorithm. This paper provides a gradient-based optimization algorithm, proof of convergence to solutions, and convergence rates of the proposed algorithm. It is shown that agents' initial priorities influence the convergence rate of the proposed algorithm and that these initial choices affect its long-run behavior. Numerical results performed with different numbers of agents illustrate the performance and efficiency of the proposed algorithm.
\end{abstract}
%\subclass{49J53 \and  49K99 \and more}

%All acknowledgements should be placed in the back of the paper after Conclusions..

\section{Introduction}

{O}{ver} the last two decades, multi-agent systems have attracted significant interest  \cite{Qin16}-\cite{Wan16}. In particular, the study of the consensus problem, where agents have to agree on a common value, has been motivated by emerging applications such as formation control \cite{Oh15}. The consensus problem has been extended to multi-agent optimization, i.e., agents collectively work towards minimizing a sum of objective functions by minimizing a local objective and repeatedly averaging their iterates to reach agreement on a final answer. One common approach in problems with many objectives is optimizing their sum with each agent independently optimizing only one of the objective functions \cite{Ned09}-\cite{Ned01}. However, optimizing the sum carries an implicit decision about the problem formulation, namely that all the objectives have the same priority and that all agents agree on these priorities. 

Equal prioritization among functions represents a special case of a multi-objective problem, and applications in which objectives may have different importance are easy to envision. For instance, in a fleet of self-driving cars, agents may have different priorities in trajectory planning such as minimizing fuel usage vs. travel time, or in a collection of smart buildings, agents may have different preferences regarding the management of their energy \cite{Byungchul13}.

A large body of work on multi-objective optimization to solve problems of this kind has emerged for centralized cases. The Tchebycheff method, the weighting method, and the $\epsilon$-Constraint method \cite{Mie99} are examples of algorithms for centralized multi-objective optimization problems. More algorithms of this category are surveyed in \cite{Mie99}\cite{Sia04}. Such algorithms explore the Pareto optimal set using different prioritizations of the objective functions of the problem. With regard to these techniques, minimizing the sum of objective functions leads to a single element of the Pareto Front. Further exploring this front can provide additional optimal solutions in different senses. For multi-agent systems, exploring the Pareto Front would provide a larger range of operating conditions for systems based on agents' needs, which can be encoded in heterogeneous weights on objectives. To the best of our knowledge, such methods remain largely unexplored in a multi-agent context. 

This paper proposes a distributed algorithm for multi-agent multi-objective set-constrained problems, and the proposed algorithm enables the exploration of the Pareto Front. In particular, a team of $m$ agents optimizes the weighted sum of convex cost functions $f(x)=\sum_{i=1}^m w_if_i(x)$, where agent $i$ minimizes $f_i$. A common convex set constrains the agents. At the beginning of the optimization process, agents have an initial vector of priorities encoded as weights and an initial vector of decision variables. The proposed algorithm performs four steps at each iteration: \textit{i)} agent $i$ updates its vector of priorities using those received from other agents in the network, \textit{ii)} the vectors of priorities are used to generate the matrix of information weights for the decision variable update, \textit{iii)} agent $i$ updates its vector of decision variables with the generated matrix and the decision variables received from its neighbors, and \textit{iv)} agent $i$ takes a gradient descent step and projects its estimates on the constraint set. 

The proposed algorithm belongs to a class of averaging-based distributed optimization algorithms, e.g., \cite{Ned09}\cite{Ned10}-\cite{Liu15}.  The existing literature considers predominantly problems with doubly-stochastic weights on agents' information exchanges. Indeed, many works rely on the doubly-stochasticity assumption in their model to provide convergence rates and proofs of convergence \cite{Ned09}\cite{Ned10a}\cite{Ned10}\cite{Zha14}-\cite{Bia11}. Computing the infinite product of doubly-stochastic matrices simplifies the analysis of agents' computations, and there exist several rules that ensure the information matrix is doubly-stochastic, such as Metropolis-based weights \cite{Xia07} and the equal-neighbor model \cite{Ols11}\cite{Blo05}. These rules restrict communication among agents and do not allow agents to individually prioritize information received from other agents in the network. In addition, these rules require coordination among agents to selected admissible information weights, which can be difficult to achieve if communicating is difficult or costly. The proposed algorithm addresses the limitations related to the doubly-stochasticity assumption in addition to giving agents increased flexibility in their choices. In particular, the following aspects distinguish our algorithm from the existing literature: 

\begin{itemize}
    \item Agents independently prioritize the information received from their neighbors. The sum of each agent's preferences must be 1. While individual agents can easily ensure that their preferences sums to 1, this implies that agents do not have know or consider other agents' preferences. Therefore,  preferences of all agents for a particular objective function need not to sum to 1.  
    \item This independence regarding prioritization of objective functions leads to a network-level information exchange matrix that is (non-doubly) stochastic.
    \item While the agents are reaching an agreement on preferences, they explore the Pareto Front of objective functions. This front exploration leads to optimal solutions in different senses, which provides broader operating conditions for systems in conformity with agents' needs/preferences. 
\end{itemize}

 Because of these distinctions, new theoretical analysis is required to ensure algorithm convergence.  In this paper, the proposed algorithm operates over an undirected graph with time-varying weights, and the constraint set is the same for all agents. Theoretical analysis shows that the proposed algorithm drives agents to a common solution. Agents simultaneously reach an agreement on their preferences and compute the optimum with respect to these agents' preferences. Also, we develop convergence rates for the proposed algorithm, which are shown to be significantly influenced by agents' preferences. 
Numerical simulations show the convergence of the proposed algorithm to the optimal solution along with its convergence rate. The agents' agreement on preferences is also illustrated. Simulations further show that agents' initial preferences directly influence the final results of their computations. 
This paper is an extension of \cite{Blondin2020a} and it adds proof of convergence and convergence rates, in addition to new simulation results. 

The rest of the paper is organized as follows. Section \ref{Graph theory and multi-agent interactions} presents background on graph theory and multi-agent interactions. The multi-agent optimization model and the proposed distributed optimization algorithm are provided in Section \ref{Multi-agent optimization model}. Section \ref{Convergence of the proposed algorithm} provides proofs of convergence and convergence rates of the proposed algorithm. Section \ref{Numerical results} presents numerical results, and Section \ref{Conclusion} concludes the paper.

\section{Graph theory and multi-agent interactions}
\label{Graph theory and multi-agent interactions}

In this paper, agents' interactions are represented by a connected and undirected graph $G = (V,E)$, where $V= [m]:=\{1,2, \ldots, m\}$ is the set of agents and $E  \: \subset V \times  V$ is the set of edges. An edge exists between agent $i$ and $j$, i.e., $(i,j) \in E$, if agent $i$ communicates with agent $j$. By convention, $(i,i) \notin E $ for all $i$. The degree of agent $i$ is the total number of agents that agent $i$ communicates with, denoted $\textnormal{deg}(i)$. The degree matrix, denoted $\Delta(G)$,  is a  diagonal $n \times n$ matrix, with $\textnormal{deg}(i)$ on its diagonal for $i=1, \ldots, n$. The maximum vertex degree of $G$ is $ \Delta_{max} = \underset{i \in [m] }{\max} \;\: \textnormal{deg}(i)$.

 The adjacency matrix is an $n \times n$ matrix denoted $H(G)$, where $h^i_j$ is the entry in the $j$-th row and $i$-th column, defined as 
 \[
 h^i_j = \begin{cases} 1 &\mbox{if } (i,j) \in E \\
0 & otherwise \end{cases}.\]

\noindent Since $G$ is an undirected graph without self-loops, $H(G)$ is symmetric with zeros on its main diagonal. The Laplacian matrix associated with G is also symmetric, and is defined as 

\begin{equation}
 L(G) = \Delta(G) -H(G).   
\end{equation}

In this paper, we consider an arbitrary graph $G$, and, because $G$ is unambiguous, we will simply write its Laplacian as $L$. 

\section{Multi-agent optimization model}
\label{Multi-agent optimization model}

In this section, we formally define the class of problems to be solved. Then we propose a multi-objective multi-agent update law for solving them. 

\subsection{Problem Formulation}

In this paper, problems in which agents minimize a prioritized sum of convex objective functions are considered. Agent $i$ minimizes only the function $f_i$, about which we make the following assumption. 

\begin{assumption} For all $i \in \{1, \ldots, m\}$, the function $f_i : \mathbb{R}^n \rightarrow \mathbb{R} $ is continuously differentiable and convex. \hfill  $\triangle$
\label{ConvexAssump}
\end{assumption} 

All agents' decision variables are constrained to lie in the set $X$, about which we assume the following. 

\begin{assumption} X is non-empty, compact, and convex. \hfill  $\triangle$
\label{Xassump}
\end{assumption} 

\noindent We next consider the following  optimization problem. 
\begin{problem} \text{Given functions} $\{f_i\}_{i \in \{1, \ldots, m\}}$ satisfying \textit{Assumption} \ref{ConvexAssump},
\begin{equation}
\underset{}{\text{minimize}} \; \sum_{i=1}^m w_if_i(x),
\label{Problem}
\end{equation}
\begin{align*}
\underset{}{\text{subject to}} \; x \in  X
\end{align*}
\label{Opt_prob}
\end{problem}

\noindent where $x$ is the vector of decision variables, $w_i$ is a priority assigned to $f_i$, $\sum_{i=1}^m w_i = 1$, and $ 0 < w_i < 1 $ for all $i$. Agent $i$ knows only its objective function $f_i$.  The constraint set $X$ is identical for all agents.  \hfill $\diamond$ 
\\

For centralized problems, the priorities $\{w_i\}_{i \in \{1, \ldots, m\}}$ are fixed. Therefore, a standard convex optimization method could solve Problem 1 in a centralized way. However, for decentralized cases, agents may choose different priorities. Agent $i$ may choose $\{w_l^i\}_{l \in \{1, \ldots, m\}}$ while agent $j$ chooses $\{w_l^j\}_{l \in \{1, \ldots, m\}}$, with $w_l^i \neq w_l^j$ for all $l$. 

As this occurs, these priorities provide each agent with the flexibility to have preferences. For instance, mobile autonomous agents generating a trajectory may want to optimize fuel usage and travel time, and each agent can prioritize these two objectives according to their own needs. If agents' priorities differ, agents are solving different problems because they minimize different overall objective functions. As a result, reaching a common solution requires devising an optimization algorithm and driving agent priorities to a common value. 

Changing agents' priorities from their initial values implies that no single agents' preferences are obeyed exactly. However, the net change across all agents can be done fairly. One such way is to drive all agent's priorities to their average value. While one could envision first computing the average priorities and then optimizing, this is undesirable because it requires solving two separate problems sequentially. Instead, we devise an update law that drives agents to a common solution by interlacing optimization steps with priority averaging steps. Also, this interlacing enables agents to continuously modify their preferences based on the task at hand.

\subsection{Proposed Update Law}

\noindent At iteration $k$, agent $i$ updates its priority vector $w^i$. Agent $i$ assigns a priority to all agents (corresponding to the objective function updated by that agent), even though agent $i$ does not communicate with all agents. This provides agent $i$ with a way to influence all final priorities, and, as will be shown below, affect the final results agents attain. Agent $i$ also updates its decision vector $x^i$ by adding the weighted estimates received from its neighbors, then minimizing its objective function $f_i$ through gradient step, and then projecting its new estimate on its constraint set $X$. We have

\begin{equation}
    w^i(k+1) =  w^i(k) + c \sum_{j=1}^n h_i^j (w^j(k) - w^i(k))
  \label{w_i(k+1)}
\end{equation}

\begin{equation}
    a_j^i(k+1) =  q^i_j w^i_j(k) + \sum\limits_{\substack{j=i \\ j=1}}^m w_j^i \Tilde{q}^i_j
  \label{a_i(k+1)}
\end{equation}

\begin{equation}
    v^i(k) =  \sum_{j=1}^n a^i_j(k)x^j(k)
  \label{v_i(k)}
\end{equation}

 \begin{equation}
x^i(k+1) = P_{X}[v^i(k) - \alpha_kd_i(k)],
\label{x^i(k+1)}
\end{equation}

\noindent where $0 \:< \: c \: < \: 1/\Delta_{max}$, $h_j^i(k)$ is the j$^{th}$  i$^{th}$ entry of $H(G)$,  $a_j^i(k)$ is the weight that agent $i$ assigns to the data provided by agent $j$ at iteration $k$, $q^i_j$ is the j$^{th}$ i$^{th}$ entry of $Q$, where $Q=H+I$ and $I$ is the identity matrix, $\Tilde{q}^i_j = 1-q^i_j$, $\alpha_k$ is the gradient step size for all agents at time $k$, $P_{X_i}$ is the projection operation, and $d_i$ is the gradient vector of agent $i$ at $x^i(k)$. Formally, $d_i(k) = \nabla f_i(x^i(k))$. 

The equivalent network-level representation of \eqref{w_i(k+1)}-\eqref{x^i(k+1)} is

\begin{equation}
    W(k+1) =  PW(k),
  \label{ConsensusEq2}
\end{equation}

\noindent where $P = I - cL(G)$ and $W(k)$ is the column matrix of agents' priorities, along with $A(k)$ in which its column vectors are $a^i(k)$ for $i=\{1, \ldots, m\}$, and $V(k)$ in which its column vectors are $v^i(k)$ for $i=\{1, \ldots, m\}$.

 \begin{equation}
A(k) = Q \circ W(k) + (W(k)\circ \Tilde{Q})J\circ I, 
\label{Acomputation}
\end{equation}

\begin{equation}
    V(k) =  A(k)X(k),
  \label{v_i(k)}
\end{equation}

 \begin{equation}
X(k+1) = P_{X_c}[V(k) - \alpha_kD(k)].
\label{x^i(k+1)}
\end{equation}

where $\circ$ denotes the Hadamard product, $P_{X_c}$ is the projection operation that projects each column of $V(k) - \alpha_kD(k)$ individually, $D(k)$ is the column matrix of $d^i(k)$ for $i=\{1, \ldots, m\}$, and $X(k)$ is the column matrix of $x^i(k)$ for $i=\{1, \ldots, m\}$ . In line with the multi-objective optimization concept, our algorithm uses the priority vectors, $w^i$ for $i=\{1, \ldots, m\}$, to quantify the importance of information received to update $x^i$ for $i=\{1, \ldots, m\}$. Contrary to most existing works, the $A(k)$ matrix is a function of $W(k)$ and $A(k)$ matrix is stochastic, instead of doubly-stochastic. This occurs because an agent can ensure that its weights sum to 1, though different agents' weights for a particular objective need not to sum to 1. This implies that $A(k)$'s row sums need not to equal 1.

\noindent In \eqref{Acomputation}, $Q \circ W(k)$ computes the Hadamard product between $Q$ and $W(k)$, where the resulting matrix contains $w_j^i(k)$ for $(i,j) \in E$, $w_i^i(k)$ for all $i$, and the remaining terms are set to zero. Therefore, if agent $i$ does not communicate with agent $j$, a zero is assigned to that agent. Regarding the second term of \eqref{Acomputation}, $(W(k)\circ \Tilde{Q})J\circ I$ creates a diagonal matrix, where the diagonal terms are the sum of each row. The first term summed to the second term  in \eqref{Acomputation} means that agent $i$ assigns to itself the weights $w_j^i$ if $(i,j) \notin E$ and assigns a zero value to the entries of the \textit{i}-th row and \textit{j}-th colum for $(i,j) \notin E$.

The next lemma pertains to the weights of the $A$ matrix and the communication between agents. 
\begin{lemma} Since $a^i_j(k)$ is obtained from \eqref{w_i(k+1)} and \eqref{a_i(k+1)}, we have 
\begin{enumerate}
    \item  $a^i_i(k) \geq \underset{j \in [n]}{\min} \; \underset{i \in [n]}{\min} \: w^i_j(0)$ for all $k \geq 0$ and all $i$.
    \item $a^i_j(k) \geq \underset{j \in [n]}{\min} \; \underset{i \in [n]}{\min} \: w^i_j(0)$ for all $k \geq 0$ and all $(i,j) \in E$.
    \item $a^i_j(k) = 0$ for all $k$ if $(i,j) \notin E$.
\end{enumerate}
\label{Lemma_etaA}
\end{lemma} 

\begin{proof}
See Appendix.

\end{proof}

\noindent To simplify the forthcoming development, Eq. \eqref{x^i(k+1)} can be written as follows \cite{Ned08}, 

 \begin{equation}
x^i(k+1) = v^i(k) - \alpha_kd_i(k) + \phi^i(k),
\label{x^i(k+1)_2}
\end{equation}

 \begin{equation}
\phi^i(k) = P_{X}[v^i(k) - \alpha_kd_i(k)] - v^i(k)+\alpha_kkd_i(k).
\label{phi^i(k+1)}
\end{equation}

\noindent For all $i$ and for all $k$ and $s$ where $k > s$, the above equivalent form allows us to express the decision variable update over time as:  

\begin{equation}
\begin{split}
        x^i(k+1)=\sum_{j=1}^m [\Phi(k,s)]_j^ix^j(s)-\sum_{r=s}^{k-1}\sum_{j=1}^m[\Phi(k,r+1)]_j^i\alpha_rd_j(r) \\ - \alpha_kd_i(k) + \sum_{r=s}^{k-1}\sum_{j=1}^m[\Phi(k,r+1)]_j^i
\phi^j(r)+\phi^i(k), 
\end{split} 
\label{x^i(k+1)_k_s} 
\end{equation}

\noindent where the transition matrix $\Phi(k,s) = A(k)A(k-1), \ldots, A(s)$ \cite{Ned09}.

\section{Convergence of the proposed algorithm}
\label{Convergence of the proposed algorithm}

This section provides the convergence analysis for the proposed algorithm \eqref{w_i(k+1)}-\eqref{x^i(k+1)}. 

The following well-known lemma confirms that the priority update \eqref{w_i(k+1)} does indeed compute average priorities. 

\begin{lemma} $\lim_{k\to\infty} w^i(k) = \overline{w} = \sum_{j=1}^m \dfrac{w^j(0)}{m}$ for $j=1, \ldots, m$.  At the network level, $W(k)=\overline{W}$, where ,$\overline{W} = \mathbbm{1} \overline{w}^\intercal$.
\label{Lemma_overline{w}}
\end{lemma}  

\begin{proof}
See \cite{Olf07}\cite{Khi07}.
\end{proof}

From Assumption \ref{ConvexAssump}, the gradient is continuous and from Assumption \ref{Xassump} $X$ is compact. Therefore, we have $||d_i(k)|| \leq L$ for all $i$. From that statement, Lemma \ref{phiBounded} follows.  

\begin{lemma}
The errors $\phi^i(k)$ satisfy $||\phi^i(k)|| \leq \alpha_kL$ for all $i$ and $k$.
\label{phiBounded}
\end{lemma}
\begin{proof}
See \cite{Ned08}.
\end{proof}

The next Lemma describes the convergence behavior of $\Phi(k,s)$. 
\begin{lemma}
From Lemma \ref{Lemma_etaA}, the convergence of $\Phi(k,s)$ is geometric according to 
\begin{equation}
    |[\Phi(k,s)]^j_i - \gamma ^j(s)| \leq C \beta^{k-s},
\end{equation}

\noindent where $B_0 = m-1$, $m$ is the number of agents, $\gamma^j(s) = \lim_{k \rightarrow \infty} [\Phi(k,s)]^j_i$,  $C=2\dfrac{(1 + \underset{j \in [m]}{\min} \; \underset{i \in [m]}{\min} \: w^j_i(0) ^{-B_0})}{1 - \underset{j \in [m]}{\min} \; \underset{i \in [m]}{\min} \: w^j_i(0)^{B_0}}$, and $\beta=(1 - \underset{j \in [m]}{\min} \; \underset{i \in [m]}{\min} \: w^j_i(0)^{B_0})^{1/B_0}$.  
\label{ConvergenceOfPhi}
\end{lemma}

\begin{proof}
 See Lemma 3 and Lemma 4 in \cite{Ned09} and Lemma \ref{Lemma_etaA} above.
\end{proof} 

\noindent To prove the convergence results, we use the following lemmas \cite{Ned08}. 

%based on the property of of the infinite sum of products of the components of two sequences
\begin{lemma}
Assume that $0 < \rho < 1$, \{$\lambda_k$\}$_{k \in \field{N}}$ be a positive scalar sequence, and $\lim_{k\to\infty} \lambda_k=0$. Then, 
\begin{align*}
    \lim_{k\to\infty} \sum_{l=0}^k\rho^{k-l}\lambda_l.
=0\end{align*}

\noindent Moreover, if $\sum_k^{\infty}\lambda_k < \infty$, we have
\begin{align*}
    \sum_{k=1}^{\infty} \sum_{l=0}^k\rho^{k-l}\lambda_l < \infty \end{align*}.
\label{SummationLemma}
\end{lemma}

\begin{proof}
See proof for Lemma 7 in \cite{Ned08}.
\end{proof}

\begin{lemma}
Assume that X is a nonempty closed convex set in $\mathbb{R}^n$. Thus, we obtain for any $x \in \mathbb{R}^n $, $|| P_X[x]-y||^2 \leq ||x-y||^2-||P_X[x]-x||^2$ for all $y \in X$. 
\label{norm(projection-y)}
\end{lemma}

\begin{proof}
See proof for Lemma 1(b) in \cite{Ned08}.
\end{proof}

\begin{lemma}
Let {$x^i(k)$} be generated by \eqref{v_i(k)}-\eqref{x^i(k+1)}. We have for any $z \in X$ and all $k\geq 0$, 
\begin{align*}
    \sum_{i=1}^m|| x^i(k+1)-z||^2 \leq \sum_{i=1}^m\sum_{j=1}^m a_j^i(w^i(k))||x^j(k)-z||^2 + \\ \alpha_k^2\sum_{i=1}^m||d_i(k)||^2 -2\alpha_k\sum_{i=1}^m
(f_i(v^i(k))-f_i(z))-\sum_{i=1}^m||\phi^i(k)||^2.\end{align*}
\label{norm(x-z)}
\end{lemma}

\begin{proof} See Appendix.
\end{proof}

\noindent The following lemma demonstrates that disagreements between agents go to 0, namely that $|| x^i(k)-x^j(k)||$ as $k\rightarrow \infty$. To assess agent disagreements, we consider agents' disagreements with the average of their decision variables,  

\begin{equation}
    y(k) = \dfrac{1}{m}\sum_{j=1}^mx^j(k). 
\end{equation}

\noindent In view of \eqref{v_i(k)} and \eqref{x^i(k+1)_2}, we have

\begin{equation}
      y(k+1)  = \dfrac{1}{m} \sum_{i=1}^m \sum_{j=1}^m a^i_j(w^i(k))x^j(k) - \dfrac{\alpha_k}{m}\sum_{i=1}^m d_i(k) + \dfrac{1}{m}\sum_{i=1}^m\phi^i(k).
      \label{y(k+1)}
\end{equation}

\begin{lemma} Let the algorithm generate iterates of {$x^i(k)$} by the algorithm \eqref{v_i(k)}-\eqref{x^i(k+1)} and consider $\{y(k)\}_{k \in \field{N}}$ defined in \eqref{y(k+1)}. \\

(a) If the stepsize is decreasing such as $\lim_{k\rightarrow \infty} \alpha_{k} = 0$, thus \\
 
 \begin{align*}
      \lim_{k\rightarrow \infty}||x^i(k)-y(k)||=0 \text{\:for all \:} i.
 \end{align*}

(b) If $\sum_{k=1}^{\infty} \alpha_{k}^2 < \infty $ therefore
\begin{align*}
\sum_{k=1}^{\infty} \alpha_k||x^i(k)-y(k)|| < \infty \text{\:for all \:} i. 
\end{align*}
\label{lemma_(norm(x-y))}
\end{lemma}
\begin{proof}
 See Appendix.

\end{proof}

From Lemma \ref{lemma_(norm(x-y))}(a), the following theorem is obtained regarding the convergence rate of $ ||x^i(k)-y(k)||$. As it has been demonstrated that agents' disagreements go to 0, as $k \rightarrow \infty$ (Lemma \ref{lemma_(norm(x-y))}a), this theorem shows the rate to reach agreement on agents' decision variable. 

\begin{theorem}
%\begin{sloppypar}
Following Assumption \ref{Xassump}, there is an $M$ such that $\sum_{j=1}^{m}||x^j(0)|| \leq M$. Let $\epsilon > 0$ be given and let $K$ be the first time that $\alpha_k \leq \epsilon$. Let C be defined as in Lemma 4.3. Then $0< \underset{j \in [m]}{\min} \; \underset{i \in [m]}{\min} \: w^j_i(0) <1$, $\alpha$ is decreasing, and $\lim_{k \rightarrow \infty}\alpha_k = 0$, and for all $k \geq K+3$, we have\\
$ ||x^i(k)-y(k)|| \leq  2mCM((1 - \underset{j \in [m]}{\min} \; \underset{i \in [m]}{\min} \: w^j_i(0)^{B_0})^{(k-1)/B_0}) \\ + 4mCL \alpha_0 \dfrac{((1 - \underset{j \in [m]}{\min} \; \underset{i \in [m]}{\min} \: w^j_i(0)^{B_0})^{(k-K)/B_0})}{1-(1 - \underset{j \in [m]}{\min} \; \underset{i \in [m]}{\min} \: w^j_i(0)^{B_0})^{1/B_0}}
    + 4\alpha_{k-1}L  + \dfrac{4mCL \alpha_0\epsilon}{1-\beta}$.
  %  \end{sloppypar}
    \label{convRate_||x^i(k)-y(k)||}
\end{theorem}

\begin{proof} Recall \eqref{norm(x-y)} and $\beta=((1 - \underset{j \in [m]}{\min} \; \underset{i \in [m]}{\min} \: w^j_i(0)^{B_0})^{1/B_0})$:

\begin{align*}
 ||x^i(k)-y(k)||\leq  2mC\beta^{k-1}\sum_{j=1}^m ||x^j(0)|| + 4mCL \sum_{r=0}^{k-2}\beta^{k-r} \alpha_r
    + 4\alpha_{k-1}L.
\end{align*}

\noindent Then, we have

\begin{equation}
 ||x^i(k)-y(k)||\leq  2mCM\beta^{k-1} + 4mCL \sum_{r=0}^{k-2}\beta^{k-r} \alpha_r
    + 4\alpha_{k-1}L. 
    \label{norm(x-y)_2}
\end{equation}

%\noindent We next analyze $\sum_{r=0}^{k-2}\beta^{k-r} \alpha_r$. $\sum_{r=0}^{k-2}\beta^{k-r} \alpha_r$ can be rewritten as follows for $ k \geq K+3$:

%\begin{equation}
 %   \sum_{r=0}^{k-2}\beta^{k-r}\alpha_r= \sum_{r=0}^{K}\beta^{k-r}\alpha_r + \sum_{r=K+1}^{k-2}\beta^{k-r}\alpha_r.
%\end{equation}

\noindent Suppose we have an arbitrary $\epsilon > 0$ and let $K$ be defined so that $\alpha_r \leq \epsilon$ (since $\alpha_r \rightarrow 0$) for all $k \geq K + 3$. We therefore have

\begin{equation}
\begin{split}
    \sum_{r=0}^{k-2}\beta^{k-r}\alpha_r \leq & \sum_{r=0}^{K}\beta^{k-r}\alpha_r + \epsilon \sum_{r=K+1}^{k-2}\beta^{k-r}  \leq \max_{0 \leq t \leq K} \alpha_t \sum_{r=0}^{K}\beta^{k-r} + 
    \epsilon \sum_{r=K+1}^{k-2}\beta^{k-r}.
    \end{split}
\end{equation}

\noindent Because of $\sum_{r=K+1}^{k-2}\beta^{k-r} \leq \dfrac{1}{1 - \beta}$, we obtain

\begin{equation}
    \sum_{r=0}^{k-2}\beta^{k-r}\alpha_r \leq \max_{0 \leq t \leq K} \alpha_t \sum_{r=0}^{K}\beta^{k-r} +  \dfrac{\epsilon}{1 - \beta}.
\end{equation}

\noindent Similarly, since $\sum_{r=0}^{K}\beta^{k-r} \leq \dfrac{\beta^{k-K}}{1-\beta}$, we obtain for all $k \geq K+3$,

\begin{equation}
    \sum_{r=0}^{k-2}\beta^{k-r}\alpha_r \leq \max_{0 \leq t \leq K} \alpha_t \dfrac{\beta^{k-K}}{1-\beta} +  \dfrac{\epsilon}{1 - \beta}.
    \label{new_sum}
\end{equation}

\noindent Inserting \eqref{new_sum} into \eqref{norm(x-y)_2}, we get for $ k \geq K+3$,

\begin{equation}
\begin{split}
 ||x^i(k)-y(k)|| \leq  2mCM\beta^{k-1} + 4mCL \Big[ \max_{0 \leq t \leq K} \alpha_t \dfrac{\beta^{k-K}}{1-\beta} +  \dfrac{\epsilon}{1 - \beta} \Big]
    + 4\alpha_{k-1}L. 
    \end{split}
    \label{norm(x-y)_3}
\end{equation}

%\noindent Since $\epsilon$ is arbitrary, $\dfrac{\epsilon}{1 - \beta}$ can be made very small. We therefore have

%\begin{equation}
 %||x^i(k)-y(k)|| \leq  2mCM\beta^{k-1} + 4mCL \max_{0 \leq t \leq K} \alpha_t \dfrac{\beta^{k-K}}{1-\beta}
 %   + 4\alpha_{k-1}L. 
 %   \label{norm(x-y)_4}
%\end{equation}

\noindent Because $\alpha_t$ is decreasing, we obtain for all $k\geq K+3$,

\begin{equation}
||x^i(k)-y(k)|| \leq  2mCM\beta^{k-1} + 4mCL \alpha_0 \Big[\dfrac{\beta^{k-K}}{1-\beta} + \dfrac{\epsilon}{1 - \beta} \Big] + 4\alpha_{k-1}L. 
    \label{norm(x-y)_5}
\end{equation} \hfill $\square$
\end{proof}
\begin{sloppypar}
\noindent The convergence rate is affected by the value of $\beta$. Recall $\beta=(1 - \underset{j \in [n]}{\min} \; \underset{i \in [n]}{\min} \: w^j_i(0)^{B_0})^{1/B_0}$, meaning the value of $\beta$ is a function of the minimum initial priority and the number of agents. The convergence rate slows down as the minimum initial agent weight decreases and the number of agents increases. Agents should therefore carefully choose their preferences. A small  initial  priority would make the convergence rate very slow, which can  harm  algorithm performance. This suggests that agents’ priorities must be balanced with need for attaining a high-quality final result with a reasonable convergence rate. Along the same lines, an extremely large team of agents would increase the limit of the convergence rate; as the number of agents increases agents' preferences associated to objective functions tend to be smaller since agents' preferences sum to 1. 
\end{sloppypar}
\noindent Based on Lemmas \ref{norm(x-z)} and \ref{lemma_(norm(x-y))}, the next theorem presents the asymptotic convergence of the proposed algorithm. In distinction to \cite{Ned08}, it is shown that the iterates $x^i(k)$ converge to an optimal solution for an information exchange matrix $A$ that is (non-doubly) stochastic, which weights are obtained from agents' priorities \eqref{Acomputation}.  
%.  . 

\begin{theorem}
The iterates ${x^i(k)}$ are generated by \eqref{ConsensusEq2}-\eqref{x^i(k+1)} with stepsize satisfying conditions of Lemma \ref{lemma_(norm(x-y))}. Assume that the optimal solutions set $X^*$ is nonempty. Therefore, an optimal point $x^* \in X^*$ exists such that

\begin{align*}
    \lim_{k\to\infty}||x^i(k)-x^*||=0 \text{\: for all\:} i.
\end{align*}
\label{lim norm(x^i(k)-x)}
\end{theorem}

\begin{proof}
From Lemma \ref{norm(x-z)}, we have 
\begin{align*}
    \sum_{i=1}^m|| x^i(k+1)-z||^2 \leq \sum_{i=1}^m\sum_{j=1}^m a_j^i(w^i(k))||x^j(k)-z||^2 + \\ \alpha_k^2\sum_{i=1}^m||d_i(k)||^2 -2\alpha_k\sum_{i=1}^m(f_i(v^i(k))-f_i(z))-\sum_{i=1}^m||\phi^i(k)||^2.\end{align*}

\noindent Using the gradient bound and by removing the last nonpositive term on the right hand side, we get

%\begin{align*}
  %  \sum_{i=1}^m|| x^i(k+1)-z||^2 \leq \sum_{i=1}^m\sum_{j=1}^m a_j^i(w^i(k))||x^j(k)-z||^2 + \alpha_k^2mL^2\\ -2\alpha_k\sum_{i=1}^m(f_i(v^i(k))-f_i(z))\end{align*}

\begin{equation}
\begin{split}
    \sum_{i=1}^m|| x^i(k+1)-z||^2 \leq \sum_{i=1}^m\sum_{j=1}^m a_j^i(w^i(k))||x^j(k)-z||^2 + \alpha_k^2mL^2 \\ -2\alpha_k\sum_{i=1}^m(f_i(v^i(k))-f_i(y(k)))-2\alpha_k(f(y(k))-f(z)).
    \end{split} \label{sum_norm(x-z)_2} 
    \end{equation}
    
 \noindent Considering the gradient boundedness and the stochasticity of weights, we have
 
 \begin{equation}
 \begin{split}
     |f_i(v^i(k))-f_i(y(k))| \leq & L||v^i(k)-y(k)||  \leq  L\sum_{j=1}^ma_j^i(w^i(k))||x^j(k)-y(k)||.
      \end{split}
     \label{norm(f_i(v^i(k))-f_i(y(k)))}
 \end{equation}
 
 \noindent Summing \eqref{norm(f_i(v^i(k))-f_i(y(k)))} over $m$ and using it in \eqref{sum_norm(x-z)_2}, we obtain,
 
\begin{equation}
\begin{split}
    \sum_{i=1}^m|| x^i(k+1)-z||^2 \leq \sum_{i=1}^m\sum_{j=1}^m a_j^i(w^i(k))||x^j(k)-z||^2 + \alpha_k^2mL^2\\ + 2\alpha_kL\sum_{i=1}^m\sum_{j=1}^ma_j^i(k)||x^j(k)-y(k)||-2\alpha_k(f(y(k))-f(z)).
    \end{split}
    \label{sum(norm(x^i(k+1)-z))}
    \end{equation}
    
\noindent Considering $z=x^* \in X^*$, and by restructuring the terms we get,

\begin{equation}
    \begin{split}
 \sum_{i=1}^m|| x^i(k+1)-x^*||^2 + 2\alpha_k(f(y(k))-f(x^*)) \leq  \\ \sum_{i=1}^m\sum_{j=1}^m a_j^i(w^i(k))||x^j(k)-x^*||^2 + \alpha_k^2mL^2 + \\ 2\alpha_kL\sum_{i=1}^m\sum_{j=1}^ma_j^i(w^i(k))||x^j(k)-y(k)||.    \end{split}
 \label{norm}
 \end{equation}

\noindent By summing \eqref{norm} over an arbitrary window from some positive integer $K$ to $N$ with $ K < N$, we obtain,

\begin{equation}
\begin{split}
    \sum_{i=1}^m|| x^i(N+1)-x^*||^2 + 2\sum_{k=K}^N\alpha_k(f(y(k))-f(x^*)) \leq \\ \sum_{i=1}^m\sum_{j=1}^m a_j^i(w^i(K))||x^j(K)-x^*||^2 + \\ mL^2\sum_{k=K}^N\alpha_k^2 + 2L\sum_{k=K}^N\alpha_k\sum_{i=1}^m\sum_{j=1}^ma_j^i(w^i(k))||x^j(k)-y(k)||. \end{split} \label{norm(x-z*)} \end{equation}
    
\noindent With $K=1$ and $N \to \infty$ in \eqref{norm(x-z*)}, using $\sum_{k=1}^{\infty} \alpha_k^2 < \infty$ and $\sum_{k=1}^{\infty}\alpha_k\sum_{j=1}^m||x^j(k)-y(k)|| < \infty$, which is a result of Lemma \ref{lemma_(norm(x-y))}, we have

\begin{align*}
    \sum_{k=1}^{\infty}\alpha_k(f(y(k))-f(x^*)) < \infty.
\end{align*}

\noindent Because $x^j(k) \in X$ for all $j$, $y(k) \in X$ for all $k$. Given that $x^* \in X^*$, $f(y(k))-f* \geq 0$ for all $k$. As a result of this relation and the assumption that $\sum_{k=1}^{\infty} \alpha_k=\infty$, and $\sum_{k=1}^{\infty}\alpha_k(f(y(k))-f(x^*)) < \infty$, we obtain,

\begin{equation}
    \liminf_{k \to \infty}(f(y(k))-f(x^*)) = 0.
    \label{f(y(k))-f(x^*)}
\end{equation}

The forthcoming development demonstrates that agents converge to the optimal point $x^*$. The nonnegative term in left side hand of  \eqref{norm(x-z*)} can be removed. Therefore, we have

\begin{equation}
\begin{split}
    \sum_{i=1}^m|| x^i(N+1)-x^*||^2  \leq \sum_{i=1}^m\sum_{j=1}^m a_j^i(w^i(K))||x^j(K)-x^*||^2  + \\ mL^2\sum_{k=K}^N\alpha_k^2 + 2L\sum_{k=K}^N\alpha_k\sum_{i=1}^m\sum_{j=1}^ma_j^i(w^i(K))||x^j(k)-y(k)||. \end{split} \end{equation}

\noindent Given that $\sum_{k} \alpha_k^2 < \infty$ and $\sum_{k=1}^{\infty}\alpha_k\sum_{i=1}^m\sum_{j=1}^ma_j^i(w^i(k))||x^i(k)-y(k)|| < \infty$, it results that ${x^i(k)}$ is bounded for each $i$, and

\begin{align*}
    \limsup_{N \to \infty} \sum_{i=1}^m|| x^i(N+1)-x^*||^2 \leq \liminf_{K \to \infty} \sum_{i=1}^m\sum_{j=1}^m a_j^i(w^i(K))||x^j(K)-x^*||^2.
\end{align*}

\noindent This implies that the scalar sequence ${\sum_{i=1}^m||x^i(k)-x^*||}$ converges for every $x^* \in X^*$. \\

\noindent Given that $\lim_{k \to \infty} ||x^i(k)-y(k)|| = 0$ (Lemma \ref{lemma_(norm(x-y))}), $\{y(k)\}_{k \in \field{N}}$ is bounded and the scalar sequence ${|| y(k) - x^*||}$ is convergent for $x^* \in X^*$. \\

\begin{sloppypar}
\noindent Because $y(k)$ is bounded, $y(k)$  has a limit point. From \eqref{f(y(k))-f(x^*)}, we have $\liminf_{k \to \infty}f(y(k))=f^*$. Considering the previous equality and the continuity of $f$, one of the limit points of $\{y(k)\}$ must be in $X^*$, which is denoted by $x^*$. Therefore, $||y(k)-x^*||$ is convergent. Thus, $\lim_{k \to \infty} y(k) = x^*$ and $\lim_{k \to \infty} ||x^i(k)-y(k)|| = 0$, which implies that each sequence $\{x^i(k)\}$ converges to the same $x^* \in X^*$. \hfill $\square$
\end{sloppypar}
\end{proof}

From Theorem \ref{lim norm(x^i(k)-x)} and Lemma \ref{a_i(k+1)}, the following convergence rate is obtained. 

\begin{theorem}
Let $\epsilon > 0$ be given and let $K$ be the first time that $\alpha_k \leq \epsilon$. Using Lemma \ref{a_i(k+1)} and Theorem \ref{convRate_||x^i(k)-y(k)||}, we have for $s \geq K+3$, 

\begin{equation}
\begin{split}
    \sum_{i=1}^m|| x^i(k+1)-x^*||^2 \leq \sum_{i=1}^m\sum_{j=1}^m  q^i_j\omega_{max}^{(k+1)} ||x^j(s)-x^*||^2  \\ +  \sum_{r=s}^k \sum_{i=1}^m\sum_{j=1}^m q^i_j\omega_{max}^{(k+1-r)} \alpha_rL \Big[ \alpha_rL +  4mCM\beta^{r-1} \\ + 8mCL\alpha_0 \Big[\dfrac{\beta^{r-K}}{1-\beta} + \dfrac{\epsilon}{1-\beta}\Big] + 8\alpha_{r-1}L].  
    \end{split}
    \label{ConvergenceRate_xstar}
\end{equation}

%\begin{equation}
%\begin{split}
%\sum_{i=1}^m|| x^i(k+1)-x^*||^2 \leq  \sum_{i=1}^m\sum_{j=1}^m  q^i_j\omega_{max}^{(k+1)} ||x^j(s)-x^*||^2  \\ +   \sum_{r=s}^k \sum_{i=1}^m\sum_{j=1}^m q^i_j\omega_{max}^{(k+1-r)} \alpha_rL \Big[ \alpha_r + 4mCM\beta^{r-1} + 8mCL\alpha_0 \dfrac{\beta^{r-K}}{1-\beta} + 8\alpha_{r-1}L + \dfrac{2\epsilon}{1-\beta} \Big] .  
%    \end{split}
%    \label{ConvergenceRate_xstar}
%\end{equation}
\end{theorem}

\begin{proof}

From \eqref{sum(norm(x^i(k+1)-z))}, we have

\begin{equation}
\begin{split}
    \sum_{i=1}^m|| x^i(k+1)-z||^2 \leq \sum_{i=1}^m\sum_{j=1}^m a_j^i(w^i(k))||x^j(k)-z||^2 + \alpha_k^2mL^2\\ + 2\alpha_kL\sum_{i=1}^m\sum_{j=1}^ma_j^i(k)||x^j(k)-y(k)||-2\alpha_k(f(y(k))-f(z)).
    \end{split}
\end{equation}

\noindent Dropping the last negative term, we find

\begin{equation}
\begin{split}
    \sum_{i=1}^m|| x^i(k+1)-x^*||^2 \leq \sum_{i=1}^m\sum_{j=1}^m a_j^i(w^i(k))||x^j(k)-x^*||^2 + \alpha_k^2mL^2\\ + 2\alpha_kL\sum_{i=1}^m\sum_{j=1}^ma_j^i(k)||x^j(k)-y(k)||.
    \end{split}
\end{equation}

\noindent Re-arranging the terms, we have

\begin{equation}
\begin{split}
    \sum_{i=1}^m|| x^i(k+1)-x^*||^2 \leq \sum_{i=1}^m\sum_{j=1}^m a_j^i(w^i(k)) \Bigg[||x^j(k)-x^*||^2 \\ + \alpha_k^2L^2 + 2\alpha_kL||x^j(k)-y(k)|| \Bigg].
    \end{split}
\end{equation}

\noindent Define $\omega(k)= \underset{j \in [n]}{\text{max}} \; \underset{i \in [n]}{\text{max}} \: a^i_j(k)$. Therefore, the maximum value that $\omega(k)$ can take is $\omega_{max} = 1 - \underset{j \in [n]}{\min} \; \underset{i \in [n]}{\min} \: w^j_i(0))$. We therefore obtain

\begin{equation}
\begin{split}
    \sum_{i=1}^m|| x^i(k+1)-x^*||^2 \leq \sum_{i=1}^m\sum_{j=1}^m  q^i_j\omega_{max} \Bigg[||x^j(k)-x^*||^2 +\\ \alpha_k^2L^2 + 2\alpha_kL||x^j(k)-y(k)|| \Bigg].
    \end{split}
\end{equation}

\noindent Using Theorem \ref{convRate_||x^i(k)-y(k)||}, we obtain 
%\begin{equation}
%\begin{split}
 %   \sum_{i=1}^m|| x^i(k+1)-x^*||^2 \leq \sum_{i=1}^m\sum_{j=1}^m  q^i_j\omega_{max}^{(k+1)} ||x^j(s)-x^*||^2  \\ +  \sum_{r=s}^k \sum_{i=1}^m\sum_{j=1}^m q^i_j\omega_{max}^{(k+1-r)} \Big[ \alpha_r^2L^2 + 2\alpha_rL\Big[ 2mCM\beta^{r-1} +\\ 4mCL \alpha_0 \dfrac{\beta^{r-K}}{1-\beta} + 4\alpha_{r-1}L \Big]  \Big]
  %  \end{split}
%\end{equation}

%\begin{equation}
%\begin{split}
 %   \sum_{i=1}^m|| x^i(k+1)-x^*||^2 \leq \sum_{i=1}^m\sum_{j=1}^m  q^i_j\omega_{max}^{(k+1)} ||x^j(s)-x^*||^2  \\ +  \sum_{r=s}^k \sum_{i=1}^m\sum_{j=1}^m q^i_j\omega_{max}^{(k+1-r)} \Big[ \alpha_r^2L^2 + 4mLCM\alpha_r\beta^{r-1} + 8mCL^2\alpha_r \alpha_0 \dfrac{\beta^{r-K}}{1-\beta} + 8\alpha_r\alpha_{r-1}L^2 \Big]  
  %  \end{split}
%\end{equation}

\begin{equation}
\begin{split}
    \sum_{i=1}^m|| x^i(k+1)-x^*||^2 \leq \sum_{i=1}^m\sum_{j=1}^m  q^i_j\omega_{max}^{(k+1)} ||x^j(s)-x^*||^2  \\ +  \sum_{r=s}^k \sum_{i=1}^m\sum_{j=1}^m q^i_j\omega_{max}^{(k+1-r)} \alpha_rL \Big[ \alpha_rL +  4mCM\beta^{r-1} \\ + 8mCL\alpha_0 \Big[\dfrac{\beta^{r-K}}{1-\beta} + \dfrac{\epsilon}{1-\beta}\Big] + 8\alpha_{r-1}L].  
    \end{split}
    \label{ConvergenceRate_xstar}
\end{equation} \hfill $\square$

\end{proof}

\begin{sloppypar}
The convergence rate is determined by $\omega_{max}$. Since $\omega_{max} = 1 - \underset{j \in [n]}{\min} \; \underset{i \in [n]}{\min} \: w^j_i(0)$, the initial agents' weights influence the convergence rate. If the smallest initial weight is extremely small, it could be detrimental for the algorithm performance as it would slow down significantly the convergence rate. Agents should consider balancing their need for reaching a high-quality final result and reasonable convergence rate. Agents should avoid extreme difference in their highest and lowest priorities. 
\end{sloppypar}

\section{Numerical results}
\label{Numerical results}

Three simulation scenarios are run to illustrate the performance of the proposed algorithm. The numerical studies considers quadratic functions defined as, 

\begin{equation}
    f_i(x) = \frac{1}{2}x^TQ_ix + r_i^Tx + c_i,
    \label{quad_func}
\end{equation}

\noindent where ~$x \in \mathbb{R}^n$ is the decision vector, $Q \in \mathbb{R}^{n \times n}$ is a symmetric positive definite matrix, ~$r \in \mathbb{R}^n$, and $i=1, \ldots, n$. The matrix~$Q_i$ and the vector~$r_i$ and $c_i$ are generated randomly and unique for each agent. An agent $i$ knows exclusively the objective function $f_i$. The agents goal is to solve the following problem using \eqref{w_i(k+1)}-\eqref{x^i(k+1)}:

\begin{equation}
     \underset{}{\textnormal{minimize}} \sum_{i=1}^m w_if_i(x).
     \label{Problem1}
\end{equation}

\begin{align*}
    \textnormal{subject} \: to \: x \in [-1,000 \:\text{to} \: 1,000].
\end{align*}

For all scenario, the initial gradient step size is $\alpha_0$ = 0.2 and we let $\alpha_k = \dfrac{\alpha_0}{k}$. 

\subsection{First simulation scenario}

The first simulation aims to show the exploration of the Pareto Front by the algorithm. For illustrative purposes, the team has two agents and the number of decision variables is 10. Simulations with different initial agent priorities have been performed with identical initial states. Agents exchange information 100,000 times. Fig. \ref{ParetoFront} presents points on the Pareto Front obtained. 

\begin{figure}[]
\centering
   \includegraphics[width=0.85\textwidth,height=0.60\textwidth]{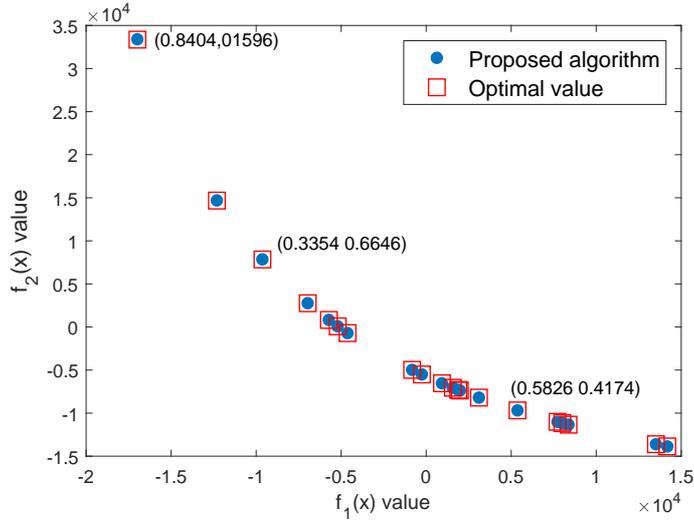}
   \caption{$f_2(x)$ in terms of $f_1(x)$ for a team of 2 agents. It can be seen that the proposed algorithm reaches various optimal solutions. Some of agents' priorities are shown. As the priority of $f_1$ diminishes, the value of $f_1$ increases and the value of $f_2$ decreases. }
    \label{ParetoFront}
\end{figure}   

Optimal solutions in different senses are reached by the agents. The network controls the exploration of the Front through the assignment of priorities, providing a wider range of "optimal" solutions in different senses.

\subsection{Second simulation scenario}
The second scenario aims to show graphically the convergence rate towards the optimal solution, i.e., \eqref{ConvergenceRate_xstar}. The team consists of three agents, $m=3$, and the decision vector has 10 variables, i.e., $n=10$. The team minimizes 10 quadratic functions as defined by \eqref{quad_func} to solve \eqref{Problem1}. The network exchanges information 100,000 times. Table \ref{ConsensusOnWeights} shows the initial agent preferences, $w_i$, and the convergence of the priority vector, $\Bar{w}$. The sum of each agent priorities equals 1, i.e., $\sum_{i=1}^3 w_i^j =1$. 

\begin{table}[!t]
\caption{Initial priorities and priorities value for consensus}
\label{ConsensusOnWeights}
\centering
\begin{tabular}{lcccc} \hline
& \multicolumn{3}{c}{Agent} & \\
 & 1 & 2 & 3 & $\Bar{w}$\\ \hline
$w_1$ & 0.3495 & 0.2232 & 0.6315  & 0.4014 \\ \hline
$w_2$ & 0.3027 & 0.3838 & 0.2494 & 0.3119 \\ \hline
$w_3$ & 0.3478 & 0.3930 &  0.1191  & 0.2866 \\ \hline
\end{tabular}
\end{table}

Table \ref{decision_vector} presents the results obtained by the proposed algorithm. The first three columns correspond to the initial decision vector of each agent. The fourth column presents the final average estimate reached by the agents, i.e., $y(k)$ for $k \rightarrow \infty$. The last column shows the optimal solution. 

\begin{table}[!t]
\caption{Result obtained by the proposed algorithm for Scenario 1}
\label{decision_vector}
\centering
\begin{tabular}{ccccc} \hline
$x_1$ & $x_2$ & $x_3$ & $\hat{x}$ & $x^*$  \\ \hline
 -728.77 & -284.03 & -981.79 & 16.72 & 16.71 \\
  -94.90 & 406.18 &  951.03 & -0.53 & -0.53 \\
  429.65 & 792.26 & -532.88 & -9.03 &   -9.02 \\
   14.82 & -360.08 &  73.41 & 5.52 & 5.52 \\
  846.91 & -797.02 &  147.99 & -4.74 & -4.74 \\
 -789.88 & 986.15 & -602.51 & 8.16 & 8.15 \\
 -285.74 & 723.87 & -584.77 &  1.01 & 1.01 \\
 -820.97 &  39.10 & -30.25 &  -5.07 &  -5.07 \\
  634.15 & -431.47 & 888.04 & -13.59 & -13.58 \\
 -352.03 & 361.52 &  -10.59    & -8.68 &  -8.67 \\  \hline
\multicolumn{3}{c}{$f(x)$} & -6.1094e+03 & -6.1094e+03 \\ \hline
\end{tabular}
\end{table}
             
The results obtained by the proposed algorithm closely approach the optimal value, $x^*$. 
Fig. \ref{Team3agents_10var_convergence_rate} presents the algorithm's convergence rate calculated with \eqref{ConvergenceRate_xstar} and $K = 1$. 

%The term $\sum_{i=1}^m|| x^i(k+1)-x^*||^2$ is always under the convergence rate bound and goes to 0 as $k\rightarrow\infty$. 
   
\begin{figure}[]
\centering
   \includegraphics[width=0.75\textwidth,height=0.45\textwidth]{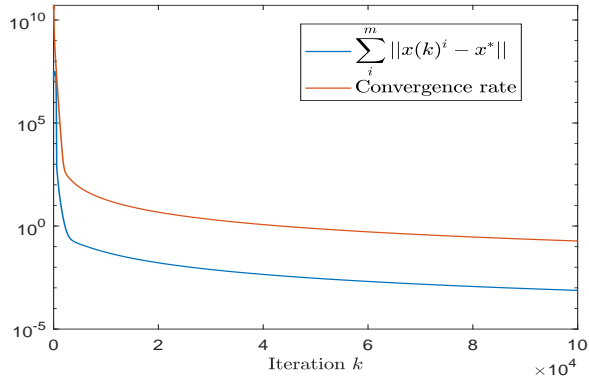}
   \caption{Convergence rate of the proposed algorithm - Team of 3 agents. The agents converge towards the optimal solution.}
    \label{Team3agents_10var_convergence_rate}
\end{figure}   

\subsection{Third simulation scenario}
 
The third simulation scenario objective is to demonstrate the proposed algorithm's efficiency on a larger team of agents and higher number of decision variables. The team consists of 100 agents with quadratic functions defined by \eqref{quad_func} of 100 variables. Therefore, the agent teams solve Problem \ref{Opt_prob} where $m=100$. The set of constraints is the same as scenario 1 and the quadratic functions are also created randomly. Fig. \ref{Team100agents_100var_convergence_rate} displays $\sum_{i=1}^m|| x^i(k+1)-x^*||^2$ over the course of the algorithm. As $k$ $\rightarrow$ $\infty$, the $\sum_{i=1}^m|| x^i(k+1)-x^*||^2$ $\rightarrow$ 0, which means the agent team approximately reach the optimal solution. Indeed, $f(x^*) = -50.11$ and $f(\hat{x}) = -49.64$. 

\begin{figure}[]
\centering
   \includegraphics[width=0.75\textwidth,height=0.45\textwidth]{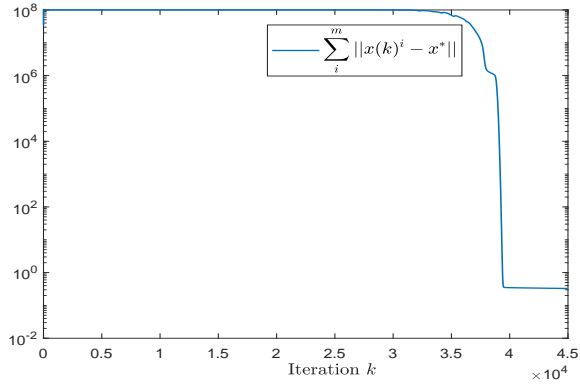}
   \caption{Convergence rate of the proposed algorithm for a team of 100 agents. During the first iterations, agents' decision variables are close to the boundary of the constraint set, which explains the values  $\sum_{i=1}^m|| x^i(k+1)-x^*||^2$ obtained during the first iterations. It takes several iterations to move agents away from the boundaries. However, once agents are far from the boundaries, agents converge quickly towards the optimal solution.}
    \label{Team100agents_100var_convergence_rate}
\end{figure} 

 A plateau followed by sharp drop is observed in the curve. The high bound for the gradient compared to the constraint set explains this phenomenon. Since the gradient can take large values, the decision variables may take large values too. However, the $x_i$ is subject to the constraint set $[-1,000 \; \textnormal{to} \; 1,000]$. Therefore, during the first iterations, most of the decision variables are projected onto the limits of the constraint set. It takes several iterations before a significant number of agents move away from the boundary of the constraint set. However, once this number is reached, the algorithm converges quickly towards the optimal solution. 

\section{Conclusions}
\label{Conclusion}
 
 In this paper, a distributed algorithm to optimize a prioritized sum of convex objective functions was proposed. The proposed algorithm allows agents to have different priorities regarding other agents' objective functions. These agents' priorities enable the exploration of the Pareto Front, which provides optimal solutions in different senses. A rule based on agents' priority generates the information exchange matrix used to update agents' estimates. In the proposed algorithm, this matrix is stochastic, whereas, in most other distributed algorithms, the information exchange matrix is doubly-stochastic. Therefore, new theoretical analyses were needed because of the difference in the network-level set-up. It has been proved that the proposed algorithm converged towards the optimal solution. Also, convergence rates were obtained, which are influenced by agents' initial weights. Numerical results illustrated the performance of the proposed algorithm. Future works include time-varying topology and implementing the algorithm on a team of robots.

\begin{acknowledgements}
Maude J. Blondin would like to thank the support of Fonds de recherche Nature et technologies postdoctoral fellowship.
\end{acknowledgements}

\appendix  %This command ends the counting of sections.
\section*{Appendix: }
This appendix contains the proofs for some lemmas presented in the paper. \\
\textit{Proof of Lemma \ref{Lemma_etaA}} \cite{Blondin2020a}

Define $ \mu(k) := \underset{j \in [m]}{\min} \; \underset{i \in [m]}{\min} \: w^j_i(k)$. Then, $W(k+1) = PW(k)$ can be expressed as

\begin{equation}
\begin{split}
\begin{bmatrix} 
p_{1}^1 &\ldots &p_{1}^m \\
\vdots & \vdots & \vdots \\
p_{i}^1 &\ldots &p_{i}^m \\
\vdots & \vdots & \vdots \\
p_{m}^1 &\ldots &p_{m}^m 
\end{bmatrix}
\quad
    \begin{bmatrix} 
\mu(k) + \delta_{1}^1(k) &\ldots &\mu(k) + \delta_{1}^m(k) \\
\vdots & \vdots & \vdots \\
\mu(k) + \delta_{i}^1(k) &\ldots & \mu(k) + \delta_{i}^m(k) \\
\vdots & \vdots & \vdots \\
\mu(k) + \delta_{m}^1(k) &\ldots & \mu(k) + \delta_{m}^m(k)
\end{bmatrix} \\
\quad =
    \begin{bmatrix} 
w_{1}^1(k+1) &\ldots &w_{1}^m(k+1) \\
\vdots & \vdots & \vdots \\
w_{i}^1(k+1) &\ldots &w_{i}^m(k+1) \\
\vdots & \vdots & \vdots \\
w_{m}^1(k+1) &\ldots &w_{m}^m(k+1) 
\end{bmatrix},
\end{split}
\end{equation}

\noindent where $\delta_{i}^j(k) = w_{i}^j(k)-\mu(k) \geq 0$ for $i,j = \{1, \ldots, n\} $. Then, we have

\begin{equation}
\begin{split}
    w_{i}^j(k+1) = & \sum_{n=1}^m p_{i}^n [\mu(k) + \delta_{n}^j(k)] = \sum_{n=1}^m p_{i}^n \mu(k) + \sum_{n=1}^m p_{i}^n\delta_{n}^j(k) \\& = \mu(k) \sum_{n=1}^m p_{i}^n  + \sum_{n=1}^m p_{i}^n\delta_{n}^j(k).
    \end{split}
\end{equation}

\noindent By definition, we know that $\sum_{m=1}^n p_{i}^m =1$. Therefore, we get 
\begin{equation}
    w_{i}^j(k+1) = \mu(k) + \sum_{n=1}^m p_{i}^n\delta_{n}^j(k).
\end{equation}

\noindent Since $\delta_{n}^j \geq 0 $ and $p_i^n \geq 0$ for $i,j,n= \{1, \ldots,m\}$,  $w_{i}^j(k+1) \geq \mu(k) =\underset{j \in [m]}{\min} \; \underset{i \in [m]}{\min} \: w^j_i(k)$ for $i,j = \{1, \ldots, m \}$ and all $k$.  This establishes that the minimum of $W(k)$ is non-decreasing and other agents cannot go below the previous minimum at the next time step. 

\begin{sloppypar}
Therefore, since \eqref{Acomputation} defines $A(k)$, the smallest non-zero element of $A(k)$, denoted $\underset{i \in [m]}{\min^+ \:}\underset{j \in [m]}{\min^+} [A(k)]^j_i$, is at least $\underset{i \in [m]}{\min \:} \underset{j \in [m]}{\min} w_i^j(k)$. This directly implies that the lower bound can be set as $\underset{j \in [m]}{\min} \; \underset{i \in [m]}{\min} \: w^j_i(0)$. \hfill $\square$

\end{sloppypar} 

\noindent \textit{Proof of Lemma \ref{norm(x-z)}} \\
From Lemma \ref{norm(projection-y)} and since $x^i(k+1)=P_{X_i}[v^i(k)-\alpha_k d_i(k)]$, we have
\begin{align*}
    ||x^i(k+1)-z||^2 \leq ||v^i(k)-\alpha_kd_i(k)-z||^2- \\ ||P_{X_i}[v^i(k)-\alpha_kd_i(k)]-v^i(k)-\alpha_kd_i(k)||^2.
\end{align*}

\noindent From the definition of $\phi^i(k)$ in \eqref{phi^i(k+1)}, the previous relation becomes,
\begin{align*}
    ||x^i(k+1)-z||^2 \leq ||v^i(k)-\alpha_kd_i(k)-z||^2-||\phi^i(k)||^2.
\end{align*}

\noindent By expanding $||v^i(k)-\alpha_k d_i(k)-z||^2$, we have
\begin{equation}
\begin{split}
    ||v^i(k)-\alpha_k d_i(k)-z||^2 = ||v^i(k)-z||^2+\alpha_k^2||d_i(k)||^2\\-2\alpha_k d_i(k)'(v^i(k)-z).
    \end{split}
     \label{norm(v-alpha d(k)-z)}
   \end{equation}

\noindent Because $d_i(k)$ is the gradient of $f_i(x)$ at $x=v^i(k)$, we obtain from convexity that 

\begin{equation}
  d_i(k)'(v^i(k)-z) \geq f_i(v^i(k))-f_i(z).
  \label{d(k)(v(k)-z)}
\end{equation}

\noindent By bringing together \eqref{norm(v-alpha d(k)-z)} and \eqref{d(k)(v(k)-z)}, we get

\begin{equation}
    \begin{split}
    ||x^i(k+1)-z||^2 \leq ||v^i(k)-z||^2 + \alpha_k^2||d_i(k)||^2 \\ - 2\alpha_k[f_i(v^i(k)-f_i(z))] - ||\phi^i(k)||^2.
    \end{split}
    \label{norm(x-z)2}
\end{equation}

\noindent Given the definition of $v^i(k)$, using the convexity of the norm squared function and the stochasticity of the $a^i(w^i(k))$, we find that

\begin{equation}
 ||v^i(k)-z||^2 \leq \sum_{j=1}^m a^i_j(w^i(k))||x^j(k)-z||^2.
 \label{norm(v-z)}
\end{equation}

\noindent It then follows from \eqref{norm(x-z)2} and \eqref{norm(v-z)} that

\begin{equation}
\begin{split}
    ||x^i(k+1)-z||^2 = \sum_{j=1}^m a^i_j(w^i(k))||x^j(k)-z||^2+\alpha_k^2||d_i(k)||^2 \\ -2\alpha_k [f_i(v^i(k))-f_i(z)] - ||\phi^i(k)||^2. 
    \end{split}
    \label{||x^i(k+1)-z||^2}
\end{equation}

\noindent By summing \eqref{||x^i(k+1)-z||^2} over $i=1, \ldots, m$, we obtain the desired relation:

\begin{align*}
\begin{split}
    \sum_{i=1}^m||x^i(k+1)-z||^2 \leq \sum_{i=1}^m\sum_{j=1}^m a^i_j(w^i(k))||x^j(k)-z||^2+ \\ \alpha_k^2\sum_{i=1}^m||d_i(k)||^2 -2\alpha_k\sum_{i=1}^m[f_i(v^i(k))-f_i(z)] - \sum_{i=1}^m||\phi^i(k)||^2. 
    \end{split}
\end{align*} \hfill $\square$

\textit{Proof of Lemma \ref{lemma_(norm(x-y))}} \\
\noindent (a) From \eqref{x^i(k+1)_k_s}, we have, 

\begin{equation}
\begin{split}
      x^i(k)=\sum_{j=1}^m [\Phi(k-1,s)]_j^ix^j(s)-\sum_{r=s}^{k-2}\sum_{j=1}^m[\Phi(k-1,r+1)]_j^i\alpha_rd_j(r) - \\ \alpha_{k-1}d_i(k-1) +  \sum_{r=s}^{k-2}\sum_{j=1}^m[\Phi(k-1,r+1)]_j^i\phi^j(r)+\phi^i(k-1).
      \end{split}
      \label{(x(k))_k_s}
\end{equation}

\noindent Using the following transition matrices 
\begin{equation}
    \Phi(k,s) = A(W(k))A(W(k-1)), \ldots, A(W(s))
\end{equation}
\noindent  and following the same logic to obtain \eqref{x^i(k+1)_k_s} \cite{Ned09}, \eqref{y(k+1)} can be re-written for all $k$ and $s$ with $k > s$ as, 

%\begin{equation}
%\begin{split}
%        y(k+1) & = \dfrac{1}{m}\sum_{i=1}^m\sum_{j=1}^m[\Phi(k,s)]_j^ix^j(s) - \\ &  \dfrac{1}{m}\sum_{i=1}^m\sum_{r=s}^{k-1}\sum_{j=1}^m[\Phi(k,r+1)]_j^i\alpha_rd_i(r)  \\ & +\dfrac{1}{m}\sum_{i=1}^m\sum_{r=s}^{k-1}\sum_{j=1}^m[\Phi(k,r+1)]_j^i\phi^i(r)-\\&\dfrac{\alpha_k}{m}\sum_{i=1}^md_i(k)+\dfrac{1}{m}\sum_{i=1}^m\phi^i(k)
%\end{split}
%\label{y(k+1)_k_s}
%\end{equation}

%\noindent The following equation is expressing \eqref{y(k+1)_k_s} as a function of $k$,

\begin{equation}
\begin{split}
        y(k) = & \dfrac{1}{m}\sum_{i=1}^m\sum_{j=1}^m[\Phi(k-1,s)]_j^ix^j(s)-  \dfrac{1}{m}\sum_{i=1}^m\sum_{r=s}^{k-2}\sum_{j=1}^m[\Phi(k-1,r+1)]_j^i\alpha_rd_i(r) + \\ & \dfrac{1}{m}\sum_{i=1}^m\sum_{r=s}^{k-2}\sum_{j=1}^m[\Phi(k-1,r+1)]_j^i\phi^i(r)-  \dfrac{\alpha_k}{m}\sum_{i=1}^md_i(k-1)+\dfrac{1}{m}\sum_{i=1}^m\phi^i(k-1).
\end{split}
\label{y(k)_k_s}
\end{equation} 

\noindent By subtracting \eqref{y(k)_k_s} from \eqref{(x(k))_k_s}, we obtain,

\begin{equation}
\begin{split}
    x^i(k)-y(k) &= \sum_{j=1}^m[\Phi(k-1,s)]^i_jx^j(s)-\dfrac{1}{m}\sum_{i=1}^m\sum_{j=1}^m[\Phi(k-1,s)]_j^ix^j(s) \\ & 
    -\sum_{r=s}^{k-2}\sum_{j=1}^m[\Phi(k-1,r+1)]_j^i\alpha_rd_j(r)+  \dfrac{1}{m}\sum_{i=1}^m\sum_{r=s}^{k-2}\sum_{j=1}^m[\Phi(k-1,r+1)]_j^i\alpha_rd_j(r) \\ & 
    -  \alpha_{k-1}d_i(k-1)+\dfrac{\alpha_{k-1}}{m}\sum_{i=1}^md_i(k-1) 
    + \sum_{r=s}^{k-2}\sum_{j=1}^{m}[\Phi(k-1,r+1)]_j^i\phi^j(r) \\ & -  \dfrac{1}{m}\sum_{i=1}^m\sum_{r=s}^{k-2}\sum_{j=1}^m[\Phi(k-1,r+1)]_j^i\phi^j(r)  
    +\phi^i(k-1)-\dfrac{1}{m}\sum_{i=1}^m\phi^i(k-1).
    \end{split}
    \label{x^i(k)-y(k)}
\end{equation}

\noindent Taking the norm of \eqref{x^i(k)-y(k)}, we get 

\begin{equation}
\begin{split}
 ||x^i(k)-y(k)||\leq \sum_{j=1}^m \left| [\Phi(k-1,s)]^i_j-\dfrac{1}{m}\sum_{i=1}^m[\Phi(k-1,s)]_j^i \right| ||x^j(s)|| \\
    + \sum_{r=s}^{k-2}\sum_{j=1}^m\left[ \left| [\Phi(k-1,r+1)]_j^i-\dfrac{1}{m}\sum_{i=1}^m[\Phi(k-1,r+1)]_j^i \right| \alpha_r||d_j(r)|| \right]\\
    + \alpha_{k-1}||d_i(k-1)||+\dfrac{\alpha_{k-1}}{m}\sum_{i=1}^m||d_i(k-1)|| \\
    + \sum_{r=s}^{k-2}\sum_{j=1}^{m}\left[ \left| [\Phi(k-1,r+1)]_j^i-\dfrac{1}{m}\sum_{i=1}^m[\Phi(k-1,r+1)]_j^i \right| ||\phi^j(r)|| \right] \\ 
    +||\phi^i(k-1)||+\dfrac{1}{m}\sum_{i=1}^m||\phi^i(k-1)||.
    \end{split}
    \label{norm(x^i(k)-y(k))}
\end{equation}

\noindent Using Lemma \ref{ConvergenceOfPhi} and for $s=0$, and $k \rightarrow \infty$, the first right-hand term of \eqref{norm(x^i(k)-y(k))}  is

\begin{equation}
\begin{split}
 ||x^i(k)-y(k)||\leq \\ 
 \sum_{j=1}^m \left[ \left| [\Phi(k-1,0)]^i_j -\gamma_j(0) \right| + \left|\dfrac{1}{m}\sum_{i=1}^m[\Phi(k-1,0)]_j^i - \gamma_j(0) \right| \right] ||x^j(0)|| \\
    + \sum_{r=0}^{k-2}\sum_{j=1}^m\left[ \left| [\Phi(k-1,r+1)]_j^i-\dfrac{1}{m}\sum_{i=1}^m[\Phi(k-1,r+1)]_j^i \right| \alpha_r||d_j(r)|| \right]\\
    + \alpha_{k-1}||d_i(k-1)||+\dfrac{\alpha_{k-1}}{m}\sum_{i=1}^m||d_i(k-1)|| \\
    + \sum_{r=0}^{k-2}\sum_{j=1}^{m}\left[ \left| [\Phi(k-1,r+1)]_j^i-\dfrac{1}{m}\sum_{i=1}^m[\Phi(k-1,r+1)]_j^i \right| ||\phi^j(r)|| \right] \\ 
    +||\phi^i(k-1)||+\dfrac{1}{m}\sum_{i=1}^m||\phi^i(k-1)||,
    \end{split}   
\end{equation}

\noindent which can be simplified as, 

\begin{equation}
\begin{split}
 ||x^i(k)-y(k)||\leq 
 2mC\beta^{k-1}\sum_{j=1}^m ||x^j(0)|| \\
    + \sum_{r=0}^{k-2}\sum_{j=1}^m\left[ \left| [\Phi(k-1,r+1)]_j^i-\dfrac{1}{m}\sum_{i=1}^m[\Phi(k-1,r+1)]_j^i \right| \alpha_r||d_j(r)|| \right]\\
    + \alpha_{k-1}||d_i(k-1)||+\dfrac{\alpha_{k-1}}{m}\sum_{i=1}^m||d_i(k-1)|| \\
    + \sum_{r=0}^{k-2}\sum_{j=1}^{m}\left[ \left| [\Phi(k-1,r+1)]_j^i-\dfrac{1}{m}\sum_{i=1}^m[\Phi(k-1,r+1)]_j^i \right| ||\phi^j(r)|| \right] \\ 
    +||\phi^i(k-1)||+\dfrac{1}{m}\sum_{i=1}^m||\phi^i(k-1)||.
    \end{split}   
\end{equation}

\noindent Similarly, using Lemma \ref{ConvergenceOfPhi}, the second right-hand term is

\begin{equation}
\begin{split}
 ||x^i(k)-y(k)||\leq  2mC\beta^{k-1}\sum_{j=1}^m ||x^j(0)|| 
    + 2mCL \sum_{r=0}^{k-2}\beta^{k-r} \alpha_r \\
    + \alpha_{k-1}||d_i(k-1)||+\dfrac{\alpha_{k-1}}{m}\sum_{i=1}^m||d_i(k-1)|| \\
    + \sum_{r=0}^{k-2}\sum_{j=1}^{m}\left[ \left| [\Phi(k-1,r+1)]_j^i-\dfrac{1}{m}\sum_{i=1}^m[\Phi(k-1,r+1)]_j^i \right| ||\phi^j(r)|| \right] \\ 
    +||\phi^i(k-1)||+\dfrac{1}{m}\sum_{i=1}^m||\phi^i(k-1)||.
    \end{split}   
\end{equation}

\noindent Using Lemma \ref{phiBounded} and the gradient bound, the third-hand right term is

\begin{equation}
\begin{split}
 ||x^i(k)-y(k)||\leq 
 2mC\beta^{k-1}\sum_{j=1}^m ||x^j(0)|| 
    + 2mCL \sum_{r=0}^{k-2}\beta^{k-r} \alpha_r
    + 2\alpha_{k-1}L \\
    + \sum_{r=0}^{k-2}\sum_{j=1}^{m}\left[ \left| [\Phi(k-1,r+1)]_j^i-\dfrac{1}{m}\sum_{i=1}^m[\Phi(k-1,r+1)]_j^i \right| ||\phi^j(r)|| \right] \\ 
    +||\phi^i(k-1)||+\dfrac{1}{m}\sum_{i=1}^m||\phi^i(k-1)||.
    \end{split}   
\end{equation}

\noindent  Using again Lemma \ref{phiBounded} and \ref{ConvergenceOfPhi}, we obtain for the last two terms, 

\begin{equation}
\begin{split}
 ||x^i(k)-y(k)||\leq 
 2mC\beta^{k-1}\sum_{j=1}^m ||x^j(0)|| 
    + 2mCL \sum_{r=0}^{k-2}\beta^{k-r} \alpha_r
    + 2\alpha_{k-1}L \\
    + 2mCL \sum_{r=0}^{k-2} \beta^{k-r}\alpha_r 
    +2\alpha_{k-1}L.
    \end{split}   
\end{equation}

\noindent We therefore obtain,

\begin{equation}
\begin{split}
 ||x^i(k)-y(k)|| \leq 
 2mC\beta^{k-1}\sum_{j=1}^m ||x^j(0)|| 
    + 4mCL \sum_{r=0}^{k-2}\beta^{k-r} \alpha_r
    + 4\alpha_{k-1}L.
    \end{split}
    \label{norm(x-y)}
\end{equation}

\noindent Since $0< \beta < 1$, $\beta^k \rightarrow 0$ as $k\rightarrow \infty$. Assuming that $\alpha_k \rightarrow 0$ and taking the limit superior, we have for all $i$,

\begin{equation}
    \limsup_{k\to\infty}||x^i(k)-y(k)|| \leq 4mCL \limsup_{k\to\infty} \sum_{r=0}^{k-2}\beta^{k-r}\alpha_r.
\end{equation}

\noindent By Lemma \ref{SummationLemma}, we have
\begin{align*}
     \lim_{k\to\infty} \sum_{r=0}^{k-2}\beta^{k-r}\alpha_r = 0.
\end{align*}

\noindent Therefore, $\lim_{k\to\infty}||x^i(k)-y(k)||=0$ for all $i$. \\

\noindent (b) By multiplying \eqref{norm(x-y)} with $\alpha_k$, we get

\begin{align*} 
 \alpha_k||x^i(k)-y(k)||\leq 
 2mC \alpha_k\beta^{k-1}\sum_{j=1}^m ||x^j(0)|| + 
 4mCL \sum_{r=0}^{k-2}\beta^{k-r} \alpha_k \alpha_r + 4 \alpha_k\alpha_{k-1}L. 
    \end{align*}
    
\noindent Using $2\alpha_k\alpha_r \leq \alpha_k^2+\alpha_r^2$ and $\alpha_k\beta^{k-1} \leq \alpha_k^2+\beta^{2(k-1)}$ for any $k$ and $r$, we obtain

\begin{align*}
   \alpha_k||x^i(k)-y(k)|| \leq 2mC\beta^{2(k-1)}\sum_{j=1}^m||x^j(0)||+ 2mC\alpha_k^2\sum_{j=1}^m||x^j(0)|| \\ + 2mCL\alpha_k^2\sum_{r=0}^{k-2}\beta^{k-r} + 2mCL\sum_{r=0}^{k-2}\beta^{k-r}\alpha_r^2 + 2L(\alpha_k^2+\alpha_{k-1}^2).
\end{align*}

\noindent Since $\sum_{r=0}^{k-2}\beta^{k-r} \leq \dfrac{1}{1-\beta}$, we have

\begin{align*}
   \alpha_k||x^i(k)-y(k)|| \leq 2mC\beta^{2(k-1)}\sum_{j=1}^m||x^j(0)||+ 2mC\alpha_k^2\sum_{j=1}^m||x^j(0)|| \\ + 2mCL\alpha_k^2\dfrac{1}{1-\beta} + 2mCL\sum_{r=0}^{k-2}\beta^{k-r}\alpha_r^2 + 2L(\alpha_k^2+\alpha_{k-1}^2).
\end{align*}

\noindent By summing from $k=1$ to $k =\infty$, we obtain

\begin{equation}
\begin{split}
   \sum_{k=1}^{\infty} \alpha_k||x^i(k)-y(k)|| \leq 2mC\sum_{k=1}^{\infty}\beta^{2(k-1)}\sum_{j=1}^m||x^j(0)||+ \\ 2mC\sum_{k=1}^{\infty}\alpha_k^2\sum_{j=1}^m||x^j(0)|| + 2mCL\dfrac{1}{1-\beta}\sum_{k=1}^{\infty}\alpha_k^2 + \\ 2mCL\sum_{k=1}^{\infty}\sum_{r=0}^{k-2}\beta^{k-r}\alpha_r^2 + 2L\sum_{k=1}^{\infty}(\alpha_k^2+\alpha_{k-1}^2).
   \end{split}
\label{sum_alpha_nrom(x-y)}
\end{equation}

In \eqref{sum_alpha_nrom(x-y)}, the first term is summable since $0 < \beta < 1$. The second and third, and fifth terms are also summable since $\sum_{k\rightarrow \infty} \alpha_{k}^2 < \infty $. By Lemma \ref{SummationLemma}, the fourth term is summable. Thus, $\sum_{k=1}^{\infty} \alpha_k||x^i(k)-y(k)|| < \infty \text{\:for all\:} i$. \hfill $\square$

%References
% BibTeX users  please use  \bibliographystyle{spmpsci_unsrt}. The option spmpsci_unsrt prints the references in JOTA format  in the order they are cited.  
%Otherwise, please use the following:

\end{document}